\documentclass[preprint,12pt]{elsarticle}

\usepackage{amsmath, amssymb, amsfonts}
\usepackage{amsthm}
\usepackage{graphicx}

\newtheorem{theorem}{Theorem}

\newtheorem{definition}[theorem]{Definition}
\newtheorem{example}[theorem]{Example}

\newtheorem{proposition}[theorem]{Proposition}
\newtheorem{remark}[theorem]{Remark}

\journal{Physica A: Statistical Mechanics and its Applications}

\begin{document}
\begin{frontmatter}

\title{Breakdown of Perturbative Expansions and \\Exact Algebraic Absorption of Finite-Size Fluctuations\\ in Statistical Mechanics}

\author{Hiroki Suyari}

\affiliation{organization={Graduate School of Informatics, Chiba University},
            addressline={1-33, Yayoi-cho, Inage-ku}, 
            city={Chiba},
            postcode={263-8522}, 
            state={Chiba},
            country={Japan}}



\begin{abstract}
In statistical mechanics, evaluating finite-size macroscopic fluctuations typically relies on Edgeworth expansions.
However, these perturbative methods append additive polynomial corrections that break down in the large deviation regime, yielding unphysical negative probabilities.
We propose a structural resolution: rather than relying on additive polynomials, we absorb finite-size skewness using a globally stable $q$-deformed framework.
By introducing a dynamic scaling law $1-q_n = O(n^{-1})$ for the nonextensivity parameter, we prove this $q$-deformed framework captures macroscopic higher-order fluctuations in independent and identically distributed (i.i.d.) systems.
Specifically, this algebraic tuning absorbs third-order skewness while guaranteeing probability density nonnegativity across the entire domain.
Furthermore, the $k$-th degree term of this $q$-logarithmic expansion corresponds to the $O(n^{1-k/2})$ asymptotic order of classical $(k+1)$-th moment Edgeworth corrections.
This correspondence functions as a stable resummation of divergent asymptotic expansions, establishing a mathematical bridge between finite-size i.i.d. fluctuations and the Tsallis statistics governing complex systems.
\end{abstract}

\end{frontmatter}

\section{Introduction}
\label{sec:introduction}

In statistical mechanics and probability theory, analyzing systems solely in the macroscopic or thermodynamic limit ($n \to \infty$, where $n$ represents the system size or the sequence length) often obscures finite-size scaling structures.
While the Central Limit Theorem guarantees a Gaussian baseline for the fluctuations of extensive variables near the mean, finite-size systems exhibit strong asymmetry (skewness) and heavy tails in the large deviation regime \cite{Seifert12, Campisi11}.
These non-Gaussian finite-size effects play a role in various physical contexts, ranging from the transient dynamics of strongly correlated systems \cite{TsallisPlastinoZheng97, LyraTsallis98} to the structural complexity at finite observation scales \cite{Halsey86}.
Therefore, characterizing these scale-dependent fluctuations is a mathematical challenge in non-equilibrium and finite-size statistical mechanics.

Conventionally, the statistical analysis of such finite-size effects relies on asymptotic expansion techniques, most notably the Edgeworth expansion (and its inverse, the Cornish-Fisher expansion).
This probabilistic approach attempts to evaluate non-Gaussian tail behaviors by appending higher-order moments (e.g., Hermite polynomials) as additive correction terms to the standard Gaussian distribution \cite{Feller71, Touchette09}.
While effective for slight deviations, this additive methodology faces structural limitations in highly nonlinear or strongly fluctuating systems, where the accumulation of higher-order perturbative terms leads to a combinatorial explosion, obscuring the physical structure of the fluctuations.

This limitation raises the following question: Rather than treating finite-size non-Gaussian fluctuations as external error penalties that must be additively corrected, is it possible to absorb (renormalize) these fluctuations by modifying the algebraic structure of the information measure itself?
To answer this question, we introduce a unified framework based on the generalized algebraic structure originating from non-extensive statistical mechanics, specifically the $q$-generalized logarithmic mapping \cite{Ts88, Tsallis09}.
Recent studies have increasingly demonstrated the efficacy of such $q$-generalized frameworks in capturing anomalous statistical behaviors and macroscopic fluctuations across various complex physical systems \cite{BOUMALI2023129134}.

We demonstrate that the non-Gaussian fluctuations inherent in finite sequences can be structurally absorbed by treating the deformation parameter $q$ as a dynamically scaling variable $q_n$ dependent on the scale $n$.
Through the mathematical formulation of the centralized $q$-generalized information density, we prove that by imposing the dynamic scaling law $1-q_n = \alpha n^{-1}$, the algebraic structure of the $q$-logarithm corresponds to the macroscopic higher-order fluctuations.
Specifically, by tuning the scaling constant to $\alpha = \frac{T}{3V^2}$ (where $V$ is the variance and $T$ is the skewness of the empirical self-information), our non-additive framework absorbs the third-order finite-size penalty, eliminating the need for classical Hermite polynomials.

The emergence of generalized non-extensive structures (Tsallis statistics) due to finite system constraints is a central topic in statistical physics.
In our preceding study \cite{Hexalogy_PartII}, we demonstrated that the finite heat capacity of an environmental bath deforms the state distribution into a $q$-exponential form.
In a complementary manner, the present paper establishes that the finiteness of the observation sequence or system size $n$ necessitates the deformation of the entropy measure itself into a $q$-logarithmic structure.
This establishment of the ``algebraic renormalization of finite-size fluctuations'' provides a mathematical foundation for the scale-dependent statistical mechanics of finite systems, offering a direct analytical pathway without relying on additive probabilistic approximations.


\section{Preliminaries: Finite-Size Fluctuations and Asymptotic Expansions}
\label{sec:preliminaries}

We briefly review the standard probabilistic tools and their current limitations when analyzing finite-size fluctuations in statistical mechanics.
In particular, we focus on the probabilistic behavior of empirical entropy and the conventional asymptotic expansions used to evaluate finite-size scaling limits.

\subsection{Shannon Empirical Entropy and Fluctuation Variance}
Consider a physical system composed of $n$ independent subsystems, or a time series of uncorrelated observations, denoted by $X^n = (X_1, X_2, \dots, X_n)$ drawn from a known probability distribution $P(X)$.
In the finite-size (or finite-time) regime, the random variable characterizing the microscopic state is the \textit{Shannon empirical entropy} (often referred to as trajectory entropy in stochastic thermodynamics).
Throughout this paper, $\ln$ denotes the natural logarithm, and all informational quantities are expressed in nats.

The Shannon empirical entropy is defined as
\begin{equation}
    s_n := -\ln P(X^n) = \sum_{i=1}^n -\ln P(X_i).
\label{eq:info_density_def}
\end{equation}
In the macroscopic or thermodynamic limit ($n \to \infty$), the Law of Large Numbers dictates that $\frac{1}{n}s_n$ converges to its expectation, the macroscopic Boltzmann-Gibbs (BG) entropy $S_{\text{BG}} = \mathbb{E}[-\ln P(X)]$.

However, for a finite system size $n$, the fluctuation of $s_n$ around its mean plays a role in non-equilibrium and complex systems.
This fluctuation is governed by the variance of the Shannon empirical entropy, termed \textit{varentropy} or information variance, defined as
\begin{equation}
    V := \mathrm{Var}[-\ln P(X)] = \mathbb{E}\left[ \left( -\ln P(X) - S_{\text{BG}} \right)^2 \right].
\label{eq:varentropy_def_prelim}
\end{equation}
Expanding this variance yields a relation for the uncentered second moment, which we will use in later algebraic developments:
\begin{equation}
    \mathbb{E}\left[ (-\ln P(X))^2 \right] = V + S_{\text{BG}}^2.
\label{eq:second_moment_relation}
\end{equation}

\subsection{Second-Order Asymptotics and Finite-Size Scaling}
While the large deviation approach characterizes the overall scaling of macroscopic fluctuations \cite{Touchette09}, the explicit boundary of the entropy fluctuations at a given tail probability $\epsilon$ (i.e., the $\epsilon$-quantile of $s_n$, denoted as $s^\star(n, \epsilon)$) is conventionally derived from the standard Central Limit Theorem.
This leads to the second-order asymptotic expansion:
\begin{equation}
    s^\star(n, \epsilon) = n S_{\text{BG}} + \sqrt{nV} Z_{\epsilon} + O(1),
    \label{eq:verdu_formula}
\end{equation}
where $Z_{\epsilon}$ is the $(1-\epsilon)$-quantile of the standard normal distribution.
The term $\sqrt{nV} Z_{\epsilon}$ represents the fluctuation width (or finite-size penalty) required to absorb the macroscopic deviations under a normal approximation guaranteed by the Central Limit Theorem.

\subsection{Limitations of the Additive Normal Approximation}
While \eqref{eq:verdu_formula} provides a baseline approximation, its reliance on the Gaussian assumption limits its accuracy when the system size $n$ is small or when the underlying distribution $P(X)$ is highly skewed.
For an asymmetric or strongly interacting underlying state (e.g., a system with extreme rare-event probabilities $p \ll 1-p$), the distribution of the Shannon empirical entropy $s_n$ exhibits strong skewness and kurtosis.
In such highly nonlinear regimes, the error induced by the normal approximation becomes significant.

To correct this discrepancy, classical probability theory and statistical physics typically employ the Edgeworth expansion \cite{Feller71} to incorporate higher-order macroscopic moments.
By inverting the Edgeworth expansion, one obtains the Cornish-Fisher expansion, which provides the third-order asymptotic expansion of the threshold as an $O(1)$ constant correction to the finite-size scaling:
\begin{equation}
    s^\star(n, \epsilon) = n S_{\text{BG}} + \sqrt{nV} Z_{\epsilon} + B(T, \epsilon) + O\left(\frac{1}{\sqrt{n}}\right),
    \label{eq:edgeworth_expansion}
\end{equation}
where $B(T, \epsilon)$ is an $O(1)$ additive correction term that explicitly depends on the third moment $T$ and the probability $\epsilon$.
If higher precision is required to capture heavier tails, further nonlinear correction terms involving the fourth moment (kurtosis) and beyond must be analytically derived and appended.

This necessity to derive and append higher-order perturbative terms highlights a structural limitation within the conventional additive normal approximation paradigm.
This mathematically demonstrates the combinatorial explosion discussed in the introduction, thereby motivating a unified constructive framework that systematically absorbs these higher-order fluctuations without requiring step-by-step additive expansions.

To make this structural limitation visually explicit, let us consider a highly asymmetric finite system, such as a biased random walk with $n=12$ independent steps and a success probability of $p=0.1$.
As illustrated in Figure \ref{fig:edgeworth_breakdown_early}, while the perturbative corrections attempt to capture the positive physical skewness ($\gamma_1 \approx 0.77$), truncating the Edgeworth expansion yields unphysical negative probabilities in the deep-tail regime ($Z \lesssim -2.4$).

\begin{figure}[htbp]
    \centering
    \includegraphics[width=0.8\linewidth]{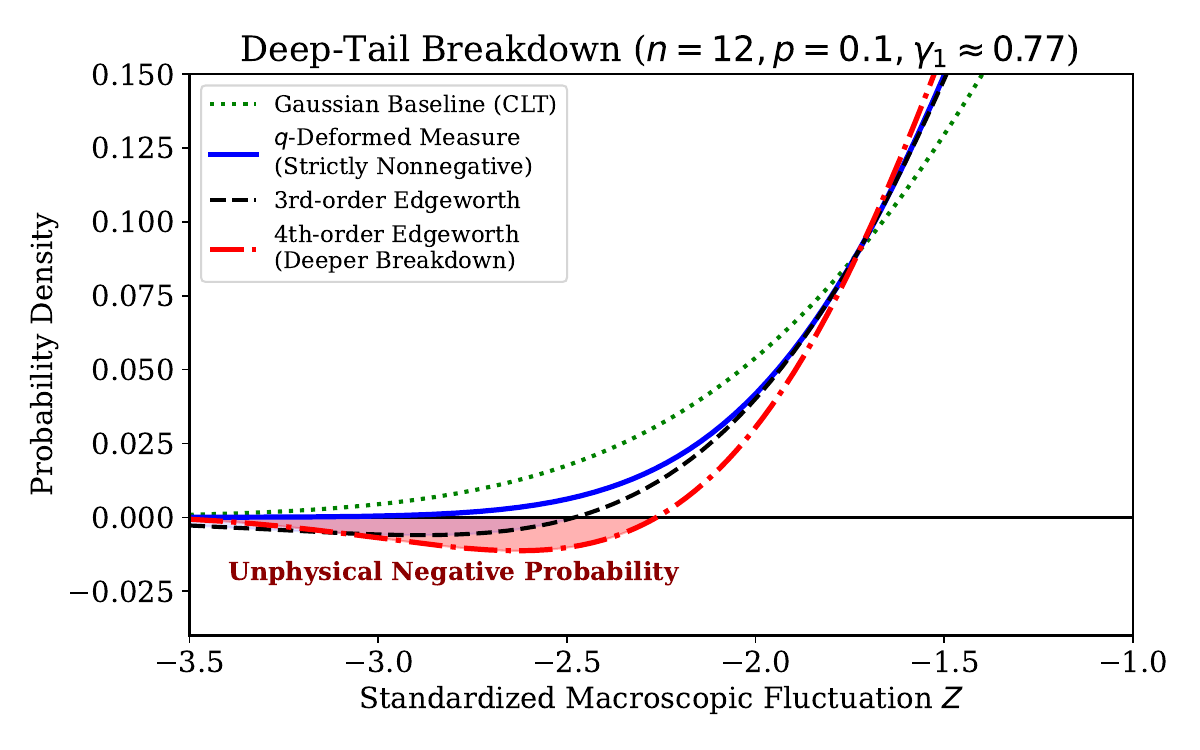}
    \caption{Breakdown of the perturbative expansion in the deep tail of a biased random walk ($n=12, p=0.1$, $\gamma_1 \approx 0.77$).
The third-order Edgeworth expansion crosses $y=0$ at $Z \approx -2.4$, yielding unphysical negative probabilities for rare macroscopic events.
Appending the fourth-order correction exacerbates this pathology. In contrast, the proposed $q$-deformed measure (detailed in Section \ref{sec:q_framework}) structurally preserves global nonnegativity.}
    \label{fig:edgeworth_breakdown_early}
\end{figure}

This zero-crossing invalidates the additive perturbative approach for large deviation analysis.
To resolve this, rather than appending polynomials, we propose absorbing this finite-size skewness structurally via a $q$-deformed algebraic structure.


\section{The $q$-Deformed Framework for Finite-Size Fluctuations}
\label{sec:q_framework}

The mathematical necessity of the $q$-logarithm in statistical mechanics can be characterized through its differential characterization. The standard logarithm uniquely linearizes the simplest proportional dynamics $dy/dx \propto y$, yielding $d(\ln y)/dx = \text{const}$, which underpins the classical additivity of macroscopic extensive variables. However, when finite-size constraints introduce nonlinear scaling into the probability measure, the canonical structural deformation of this underlying dynamic is $dy/dx \propto y^q$. This separable differential equation is uniquely linearized by the $q$-logarithm, yielding $d(\ln_q y)/dx = \text{const}$.

Based on this mathematical foundation, we formally introduce the operators of our approach. The framework relies on the $q$-logarithm originally introduced in the context of nonextensive statistical mechanics \cite{Tsallis09,Ts94}. For $x > 0$ and $q \neq 1$, the $q$-logarithm is defined as:
\begin{equation}
    \ln_q x := \frac{x^{1-q} - 1}{1-q}.
\label{eq:q_log_def}
\end{equation}
It is straightforward to verify via L'H\^{o}pital's rule that this generalized function recovers the standard natural logarithm in the Boltzmann-Gibbs (BG) limit, $\lim_{q \to 1} \ln_q x = \ln x$.
Correspondingly, the generalized macroscopic entropy (Tsallis entropy) for a probability distribution $P$ is directly defined using the $q$-logarithm as:
\begin{equation}
    S_q(P) := \sum_{x} P(x) \ln_q \frac{1}{P(x)} = \frac{1 - \sum_{x} P(x)^q}{q-1}.
\label{eq:tsallis_def}
\end{equation}

\begin{remark}[Notation on Macroscopic and Empirical Entropies]
In the statistical mechanics literature (e.g., \cite{Ts88}), the generalized macroscopic entropy is conventionally denoted by $S_q$.
To prevent ambiguity in this paper, we consistently use capital letters ($S_q$ and $S_{\text{BG}}$) to denote the macroscopic entropy as an expectation, while reserving lowercase letters ($s_q$ and $s_n$) to denote the microscopic empirical entropy (or information density) as a fluctuating random variable.
\end{remark}

To understand how this $q$-deformed framework relates to the finite-size scaling of fluctuations, we examine the local behavior of these functions near the BG limit ($q=1$).
The Taylor expansion of the $q$-logarithm around $q=1$ yields:
\begin{equation}
    \ln_q x = \ln x + \frac{1-q}{2} (\ln x)^2 + O\left((1-q)^2\right).
\label{eq:q_log_expansion}
\end{equation}
Substituting the algebraic expansion \eqref{eq:q_log_expansion} into the definition of $S_q(P)$ in \eqref{eq:tsallis_def}, and applying the expectation term by term, we obtain:
\begin{equation}
    S_q(P) = \sum_{x} P(x) \left[ \ln \frac{1}{P(x)} + \frac{1-q}{2} \left(\ln \frac{1}{P(x)}\right)^2 + O\left((1-q)^2\right) \right].
\label{eq:entropy_intermediate}
\end{equation}
The first term corresponds to the macroscopic BG entropy, $S_{\text{BG}} = \mathbb{E}[-\ln P(X)]$.
The second term involves the uncentered second moment of the Shannon empirical entropy.
By directly applying the variance relation established in \eqref{eq:second_moment_relation}, the generalized entropy mathematically decomposes as follows:
\begin{equation}
    S_q(P) = S_{\text{BG}} + \frac{1-q}{2} \left[ V + S_{\text{BG}}^2 \right] + O\left((1-q)^2\right).
\label{eq:entropy_expansion}
\end{equation}
This expansion constitutes a mathematical connection between nonextensive statistical mechanics and second-order fluctuation theory.

\subsection{System Size Extension and $q$-Empirical Entropy}
Extending \eqref{eq:entropy_expansion} to a system of size $n$ (or a sequence of $n$ observations) yields further insights into finite-size fluctuations.
To capture the finite-size deviation structurally, we introduce the raw $q$-empirical entropy as the $q$-logarithmic generalization of the standard Shannon empirical entropy $s_n$:
\begin{equation}
    s_q(X^n) := \ln_q \left[ \frac{1}{P(X^n)} \right].
\label{eq:raw_q_entropy_def}
\end{equation}

\begin{example}[Fluctuations in a Biased Random Walk]
To formalize the finite-size scaling discussed in Section \ref{sec:preliminaries}, consider a particle performing a biased random walk of $n$ independent steps, where it moves right with probability $p$ and left with $1-p$.
The microscopic state $X^n$ represents the sequence of steps, and the standard Shannon empirical entropy $s_n = -\ln P(X^n)$ quantifies the rarity of a specific realized path.
In this case, the macroscopic BG entropy and its variance scale linearly as $n S_{\text{BG}}$ and $nV$, respectively.

Applying our $q$-deformed expansion to this physical system, the generalized macroscopic entropy $S_q(P^{(n)})$ is expressed as:
\begin{equation}
    S_q(P^{(n)}) = n S_{\text{BG}} + \frac{1-q}{2} \left[ nV + n^2 S_{\text{BG}}^2 \right] + O\left((1-q)^2\right).
\label{eq:block_entropy_expansion}
\end{equation}
Here, while the $n^2 S_{\text{BG}}^2$ term arises deterministically from the macroscopic mean, the $nV$ term captures the macroscopic fluctuation (variance) of the entire trajectory.

For a finite $n$, the distribution of the particle's final position is not yet a Gaussian;
it retains an asymmetry (skewness) inherited from the bias $p$, as illustrated in Figure \ref{fig:random_walk_concept}(a).
The raw $q$-empirical entropy $s_{q}(X^n)$ is specifically designed to capture and absorb this finite-size deviation algebraically (Figure \ref{fig:random_walk_concept}(b)).

\begin{figure}[htbp]
    \centering
    \includegraphics[width=\textwidth, height=\textheight, keepaspectratio]{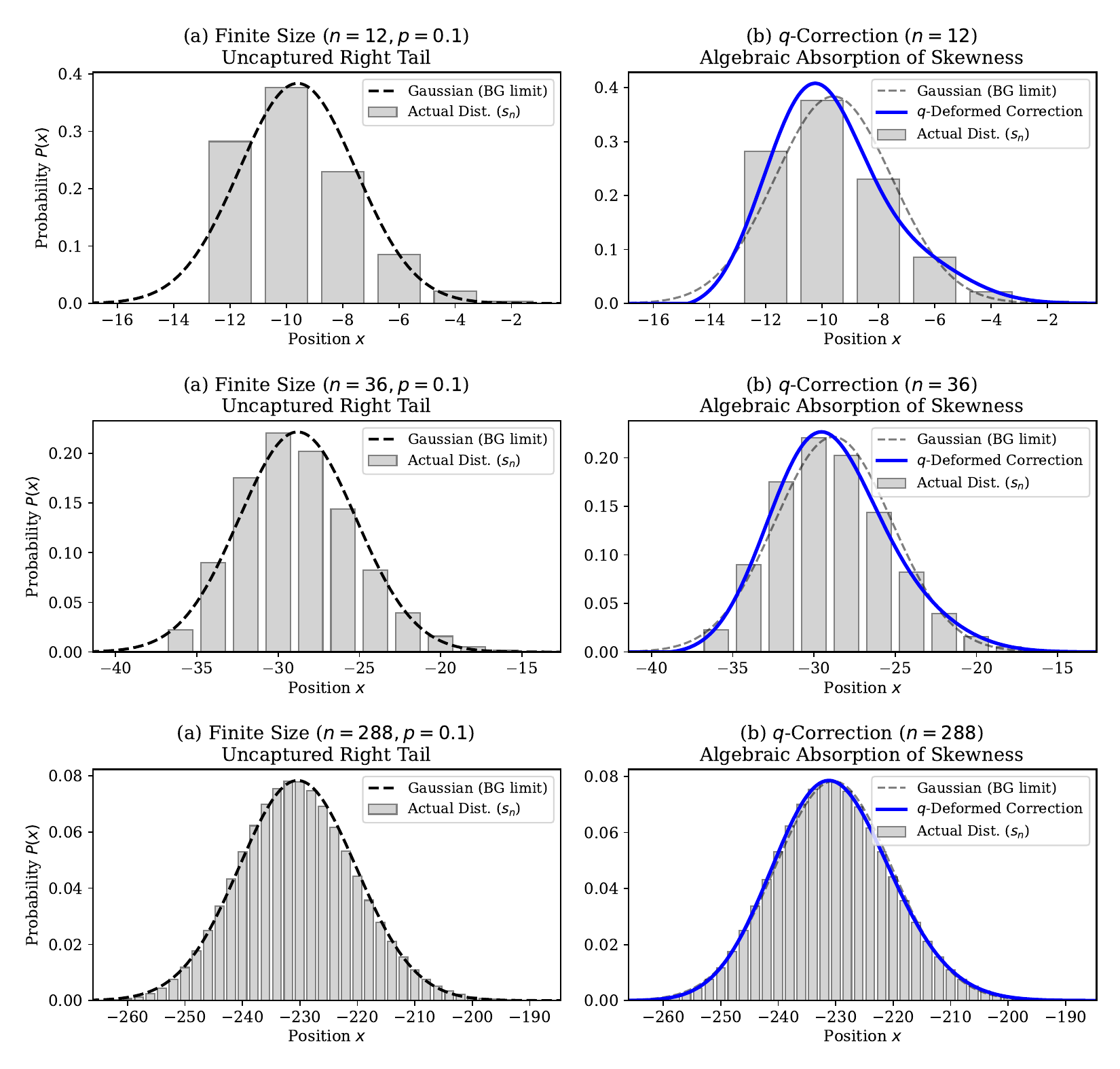}
    \caption{Finite-size skewness and $q$-deformed correction in a biased random walk ($p=0.1$).
We systematically increase the system size from $n=12$ to $n=288$ to illustrate the convergence behavior.
(a) Finite-size skewness: For finite $n$, the actual distribution deviates from the Gaussian approximation (dashed line) and exhibits a prominent right-skewed tail ($\gamma_1 > 0$), indicating that the standard Central Limit Theorem fails to fully capture the macroscopic fluctuation.
(b) Algebraic absorption of skewness: By introducing the tuning parameter $q$ governed by the dynamic scaling law $1-q_n = O(n^{-1})$, the $q$-deformed structure self-consistently absorbs and corrects this finite-size right-skewed tail (solid blue line), preserving strict nonnegativity and bridging the gap between the discrete physical reality and the asymptotic Gaussian limit.}
    \label{fig:random_walk_concept}
\end{figure}
\end{example}

\subsection{Dynamic Scaling of the Entropic Index $q$}

Conventional finite-size analysis relies on normal approximations and Edgeworth expansions to evaluate the tail bounds of $s_n$.
Rather than treating the entropic index $q$ as a fixed universal constant, we introduce a tuning parameter, $q_n$, that scales dynamically with the system size $n$.
To determine the required scaling law, we apply the $q$-logarithmic transformation to the centered fluctuation of the Shannon empirical entropy.

Let $W_n = s_n - n S_{\text{BG}}$ denote this macroscopic fluctuation, which scales as $O(n^{1/2})$ by the Central Limit Theorem.
Note that $W_n$ has zero mean ($\mathbb{E}[W_n] = 0$) and variance $\mathbb{E}[W_n^2] = nV$.
Mathematically, to understand why a specific structural centralization is required, let us examine the direct application of the $q$-logarithmic mapping.
Using the macroscopic fluctuation $W_n$, the inverse probability measure can be decomposed as $1/P(X^n) = \exp(n S_{\text{BG}} + W_n)$.
If one were to apply the raw $q$-logarithm directly to this measure, the generalized empirical entropy would take the form:
\begin{equation*}
    \ln_{q_n} \left[ \frac{1}{P(X^n)} \right] = \frac{\exp((1-q_n)n S_{\text{BG}}) \exp((1-q_n)W_n) - 1}{1-q_n}.
\end{equation*}
Under the dynamic scaling law $1-q_n = \alpha n^{-1}$, the deterministic exponential term becomes $\exp(\alpha S_{\text{BG}})$, which acts as an $O(1)$ constant multiplicative factor.
Because of this non-linear coupling, the standard expectation of this raw quantity inherently shifts the macroscopic mean away from the deterministic thermodynamic baseline $n S_{\text{BG}}$.

To conserve the deterministic macroscopic limit and isolate the purely statistical finite-size effects, we must structurally decouple the mean from the $q$-deformed fluctuation.
This is achieved by anchoring the physical baseline exactly at $n S_{\text{BG}}$ and additively appending the centralized fluctuation.
Specifically, we extract the fluctuation operator $\exp((1-q_n)W_n)$ and center it probabilistically by subtracting its expectation---which corresponds to the moment-generating function (MGF) of $W_n$.
This physical and mathematical requirement leads to the following formulation.

\begin{definition}[Algebraic Representation of Centralized $q$-Empirical Entropy]
Under the dynamic scaling law $1-q_n = \alpha n^{-1}$, the centralized $q$-empirical entropy $s^{(c)}_{q_n}(X^n)$ is defined as the structural representation that algebraically absorbs the finite-size fluctuation $W_n$:
\begin{equation}
    s^{(c)}_{q_n}(X^n) := n S_{\text{BG}} + \frac{\exp((1-q_n)W_n) - \mathbb{E}[\exp((1-q_n)W_n)]}{1-q_n}.
\label{eq:centralized_q_def}
\end{equation}
For notational simplicity in the subsequent theorems and proofs, we hereafter omit the superscript $(c)$ and denote this centralized version simply as $s_{q_n}(X^n)$ unless otherwise specified.
\end{definition}

\begin{proposition}[Dynamic Scaling Law]
\label{prop:scaling}
Let $W_n = s_n - n S_{\text{BG}}$ be the macroscopic fluctuation with variance $\mathbb{E}[W_n^2] = nV$.
In the algebraic expansion of the centralized $q$-empirical entropy, the leading quadratic fluctuation term, $\frac{1-q_n}{2} (W_n^2 - nV)$, scales asymptotically as an $O(1)$ constant if and only if the tuning parameter satisfies:
\begin{equation}
    1-q_n = O(n^{-1}).
\label{eq:q_scaling_law}
\end{equation}
\end{proposition}

\begin{proof}
We derive the polynomial behavior by expanding both the exponential and its expectation via Taylor series with respect to $(1-q_n)$:
\begin{equation}
    \exp((1-q_n)W_n) = 1 + (1-q_n)W_n + \frac{(1-q_n)^2}{2} W_n^2 + \frac{(1-q_n)^3}{6} W_n^3 + \dots
\end{equation}
and its expectation (noting that $\mathbb{E}[W_n] = 0$ and $\mathbb{E}[W_n^2] = nV$):
\begin{equation}
    \mathbb{E}[\exp((1-q_n)W_n)] = 1 + \frac{(1-q_n)^2}{2} nV + \frac{(1-q_n)^3}{6} \mathbb{E}[W_n^3] + \dots
\end{equation}

Substituting these expansions back into \eqref{eq:centralized_q_def} and dividing by $(1-q_n)$, the constant term ($1$) cancels out, yielding the centralized formulation:
\begin{equation}
    s_{q_n}(X^n) = n S_{\text{BG}} + W_n + \frac{1-q_n}{2} (W_n^2 - nV) + \frac{(1-q_n)^2}{6} (W_n^3 - \mathbb{E}[W_n^3]) + \dots
    \label{eq:q_fluctuation_def}
\end{equation}

This probabilistic centralization yields the quadratic term $(W_n^2 - nV)$.
To analyze its asymptotic contribution under the Gaussian baseline, we substitute the normalized fluctuation $W_n = \sqrt{nV}Z$ (where $Z \sim \mathcal{N}(0,1)$).
This term then maps to $nV(Z^2 - 1)$, which generates the second Hermite polynomial $He_2(Z) = Z^2 - 1$.
Since $W_n = O(n^{1/2})$, the quadratic term scales as $W_n^2 = O(n)$.

In conventional finite-size scaling formulas, however, the skewness correction (derived via the Berry-Esseen or Edgeworth expansion using this Hermite polynomial) scales as an $O(1)$ constant, because the $O(n^{-1/2})$ skewness cancels with the $O(n^{1/2})$ standard deviation.
To renormalize the $O(n)$ fluctuation down to $O(1)$ within the quadratic term $\frac{1-q_n}{2} O(n)$, the entropic index must be tuned as $1-q_n = O(n^{-1})$.
This completes the proof.
\end{proof}

\begin{remark}[Algebraic Absorption of Skewness and Kurtosis]
Physically, for a finite number of steps $n$ (such as in the biased random walk), the fluctuation distribution is not strictly Gaussian;
it exhibits asymmetric skewness and heavy tails. Proposition \ref{prop:scaling} demonstrates that the dynamic scaling $1-q_n = O(n^{-1})$ is the exact algebraic tuning required to absorb this non-Gaussian skewness directly into the $q$-logarithmic structure.
Under this scaling law, the subsequent cubic term $\frac{(1-q_n)^2}{6} (W_n^3 - \mathbb{E}[W_n^3])$ in \eqref{eq:q_fluctuation_def} scales as $O(n^{-2}) \times O(n^{3/2}) = O(n^{-1/2})$.
This matches the asymptotic order of the fourth-moment (kurtosis) correction, demonstrating that the generalized framework systematically absorbs higher-order physical fluctuations via its inherent algebraic structure.
\end{remark}


\section{Edgeworth Equivalence and Finite-Size Fluctuation Bounds}
\label{sec:proof}

In this section, we prove that the proposed $q$-deformed framework, under the dynamic scaling law $1-q_n = O(n^{-1})$, recovers the asymmetric higher-order fluctuation bounds established by conventional Edgeworth expansions, without relying on perturbative additive corrections.

\subsection{The Renormalization of Higher-Order Moments}
Unlike conventional finite-size scaling methods that append polynomial correction terms to a Gaussian baseline, our framework incorporates these higher-order physical fluctuations directly into the algebraic properties of the measure.
By adopting the centralized $q$-empirical entropy $s_{q_n}(X^n)$ as the fundamental observable, we algebraically renormalize the macroscopic fluctuations (such as skewness observed in finite-step random walks) into the operational structure of the statistical mechanics framework.

\subsection{Tail Probability Bounds under $q$-Deformation}

In the statistical mechanics of complex systems, characterizing finite-size effects requires moving beyond simple variance to evaluate the tail probabilities of macroscopic observables.
Specifically, in non-equilibrium physics and large deviation theory, it is necessary to determine a macroscopic threshold—a deterministic boundary that contains the system's fluctuating state with a high probability $1-\epsilon$.

\begin{example}[Fluctuation Boundary in a Random Walk]
Returning to our biased random walk example, suppose we want to establish a strict boundary for the Shannon empirical entropy $s_n$ (which reflects the rarity of the walker's macroscopic trajectory) after a finite number of steps $n$.
Because the true finite-size distribution is skewed (as visualized in Fig. \ref{fig:random_walk_concept}), calculating this threshold purely based on a symmetric Gaussian approximation will lead to an inaccurate estimation of the tail probabilities.
To capture the true physical boundary, we need a generalized threshold that mathematically incorporates this structural skewness.
\end{example}

To formulate this physically corrected boundary within our algebraic framework, we translate the concept of the probabilistic limit into the language of generalized statistical mechanics.

\begin{definition}[$q$-Generalized Fluctuation Threshold]
In statistical mechanics, the scaling of rare events is evaluated via tail probabilities.
We define the $q$-generalized fluctuation threshold $L_q$ as the critical boundary satisfying the tail probability under the generalized measure:
\begin{equation}
    P\left( s_{q_n}(X^n) \le L_q \right) \ge 1-\epsilon,
\label{eq:q_fluctuation_bound}
\end{equation}
where $\epsilon \in (0,1)$ is a specified tail probability for macroscopic deviations.
\end{definition}

\textit{Physical Interpretation via Fluctuation Boundaries:} 
We visualize the physical significance of this inequality as a geometric boundary problem in the macroscopic state space, as illustrated in Fig. \ref{fig:q_threshold_concept}.

\begin{figure}[htbp]
    \centering
    \includegraphics[width=\linewidth]{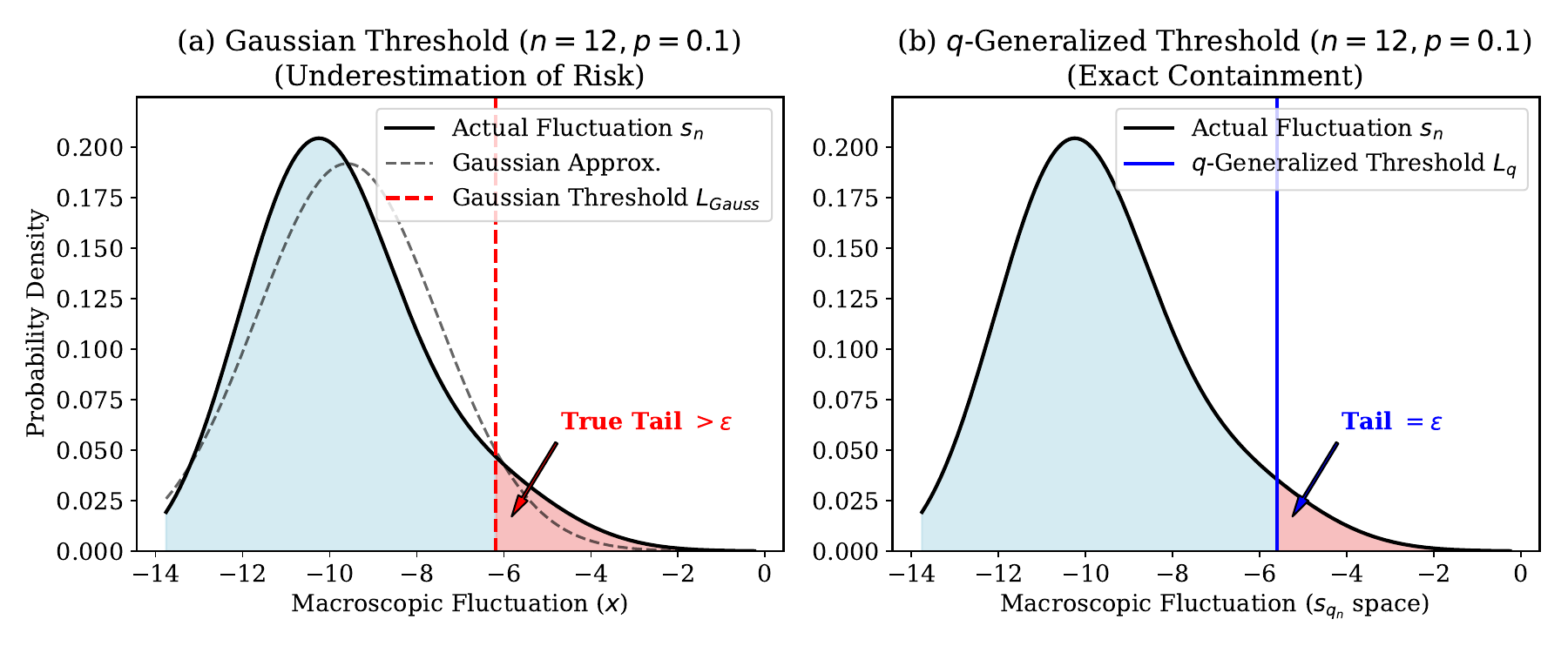}
    \caption{Geometric interpretation of the fluctuation threshold and tail probability in the presence of finite-size skewness.
(a) A threshold $L_{\text{Gauss}}$ determined by the standard Gaussian approximation fails to contain the skewed finite-size fluctuations, resulting in an anomalous tail probability that exceeds the specified tolerance $\epsilon$.
(b) The $q$-generalized threshold $L_q$ algebraically absorbs the physical skewness, establishing a strict macroscopic boundary where the probability of the bulk (safe zone) is rigorously guaranteed to be at least $1-\epsilon$.}
    \label{fig:q_threshold_concept}
\end{figure}

Here, $X^n$ represents a specific microscopic trajectory (e.g., an $n$-step random walk), and $s_{q_n}(X^n)$ is its corresponding macroscopic fluctuation.
The threshold $L_q$ acts as a deterministic macroscopic wall.
As shown in Fig. \ref{fig:q_threshold_concept}(a), if we construct this wall based solely on a symmetric Gaussian approximation, the inherent finite-size skewness causes the true anomalous tail probability (red shaded area) to leak beyond the target tolerance $\epsilon$.

Conversely, as shown in Fig. \ref{fig:q_threshold_concept}(b), the event $s_{q_n}(X^n) \le L_q$ means that the physical fluctuation is strictly contained within the probabilistically safe bulk zone (blue shaded area).
Equation \eqref{eq:q_fluctuation_bound} dictates that by using the algebraically corrected measure $s_{q_n}$, the system will fall within this wall with a high probability of at least $1-\epsilon$, isolating the extreme rare events (tail probability $\le \epsilon$).

To isolate the algebraic contribution of the $q$-deformation, we evaluate this limit by mapping the macroscopic fluctuation $W_n$ to its Gaussian equivalent $\sqrt{nV}Z$ (where $Z \sim \mathcal{N}(0, 1)$).
Recalling \eqref{eq:q_fluctuation_def}, under this normal baseline, the $q$-generalized random variable is asymptotically expressed as:
\begin{equation}
    \tilde{s}_{q_n}(Z) = nS_{\text{BG}} + \sqrt{nV}Z + \frac{1-q_n}{2} nV(Z^2 - 1) + O(n^{-1/2})
\label{eq:q_random_variable_normal}
\end{equation}
where $\tilde{s}_{q_n}(Z)$ denotes the standardized $q$-empirical entropy expressed as a function of the scaled macroscopic fluctuation $Z$.

\subsection{Algebraic Absorption of Skewness and Edgeworth Matching}

To confirm that our structural framework correctly captures discrete finite-size effects, we mathematically match its boundary condition against the classical Edgeworth expansion.

\begin{theorem}[Exact Algebraic Absorption]
\label{thm:absorption}
By evaluating the $q$-generalized threshold under the standard normal baseline and setting the tuning parameter's scaling constant to:
\begin{equation}
    \alpha = \frac{T}{3V^2},
    \label{eq:optimal_alpha}
\end{equation}
where $V$ is the fluctuation variance and $T$ is the third central moment (skewness), the resulting boundary condition coincides with the third-order Edgeworth expansion of the standard Shannon empirical entropy.
\end{theorem}

\begin{proof}
Applying the dynamic scaling law $1-q_n = \alpha n^{-1}$ to the normal-approximated random variable $\tilde{s}_{q_n}(Z)$ in \eqref{eq:q_random_variable_normal}, the threshold $L_q$ satisfying $P(\tilde{s}_{q_n}(Z) \le L_q) = 1-\epsilon$ is derived by evaluating the quantile at $Z_{\epsilon} = \Phi^{-1}(1-\epsilon)$ (where $\Phi$ is the standard normal cumulative distribution function):
\begin{equation}
    L_q = nS_{\text{BG}} + \sqrt{nV} Z_{\epsilon} + \frac{\alpha n^{-1}}{2} nV (Z_{\epsilon}^2 - 1) = nS_{\text{BG}} + \sqrt{nV} Z_{\epsilon} + \frac{\alpha V}{2} (Z_{\epsilon}^2 - 1).
\label{eq:q_limit_explicit}
\end{equation}

In classical probability theory, the third-order fluctuation bound of the Shannon empirical entropy $s_n$ is derived via the Cornish-Fisher expansion.
Recalling the additive correction $B(T, \epsilon)$ introduced in \eqref{eq:edgeworth_expansion}, its exact analytical form based on the third central moment is given by \cite{Feller71}:
\begin{equation}
    L_{\text{edge}} = nS_{\text{BG}} + \sqrt{nV} Z_{\epsilon} + \frac{T}{6V} (Z_{\epsilon}^2 - 1).
\end{equation}

Equating the $q$-deformed residual term in $L_q$ to the analytical Cornish-Fisher expansion $L_{\text{edge}}$ to enforce structural matching yields:
\begin{equation}
    \frac{\alpha V}{2} = \frac{T}{6V} \implies \alpha = \frac{T}{3V^2}.
\label{eq:alpha_identification}
\end{equation}
Thus, by tuning the algebraic deformation via $\alpha$, the generalized measure $\tilde{s}_{q_n}$ evaluated under the Gaussian baseline absorbs the true physical skewness penalty of the underlying distribution.
\end{proof}

In classical perturbative approaches, skewness is treated as an external error requiring an additive polynomial correction.
In contrast, the exact matching in Eq. \eqref{eq:alpha_identification} reveals that the dynamic tuning of $\alpha$ intrinsically embeds this physical skewness into the non-linear structure of the $q$-logarithm.
Consequently, the third-order macroscopic fluctuation is algebraically renormalized rather than additively appended, structurally preventing the emergence of negative probabilities.

This mathematical equivalence confirms that the dynamic entropic index $q_n$ acts as a structural control parameter.
By selecting $\alpha$ appropriately, the $q$-generalized empirical entropy $s_{q_n}$ absorbs the higher-order macroscopic correlations, thereby providing a self-contained, generalized thermodynamic bound for finite-size systems.


\section{Breakdown of Perturbative Approximations and the Global Stability of the $q$-Deformed Framework}
\label{sec:stability}

In the previous section, we established that the $q$-generalized fluctuation threshold mathematically coincides with the third-order Edgeworth expansion when tuned via $1-q_n = O(n^{-1})$.
We now justify the introduction of the $q$-deformed algebraic structure over the classical Edgeworth expansion.

In this section, we demonstrate that while perturbative expansions suffer from unphysical breakdowns in the deep tail, the $q$-deformed framework guarantees global physical stability and nonnegativity.

\subsection{The Unphysicality of the Edgeworth Expansion in the Deep Tail}
Conventional finite-size scaling methods treat macroscopic fluctuations as perturbative corrections to a Gaussian baseline.
For instance, the third-order Edgeworth expansion corrects the Central Limit Theorem by appending a skewness term involving the Hermite polynomial $He_2(Z) = Z^2 - 1$.
While this additive polynomial correction provides accurate local approximations near the mean (the bulk of the distribution), its structural validity breaks down when analyzing rare events (the deep tail).

Because it relies on a finite-order polynomial expansion, evaluating the Edgeworth density function far from the mean leads to oscillating behaviors.
Consequently, in the large deviation regime, the approximated probability density can become negative.
In statistical mechanics, a negative probability—or a negative probability measure for macroscopic state boundaries—is strictly unphysical, rendering the Edgeworth expansion theoretically invalid for deep-tail analysis.

\subsection{Global Nonnegativity via $q$-Deformed Algebraic Structure}
In contrast to additive perturbations, the $q$-deformed framework incorporates finite-size fluctuations through a structural deformation of the probability measure itself.
As defined in \eqref{eq:centralized_q_def}, the centralized $q$-empirical entropy $s^{(c)}_{q_n}$ and its corresponding fluctuation threshold $L_q$ are strictly constructed via the $q$-logarithmic and $q$-exponential functions.

Because the dynamic scaling law $1-q_n = \alpha n^{-1}$ acts within the algebraic structure of the entropy rather than as an external polynomial correction, it avoids the zero-crossing behavior of the Edgeworth expansion.
The $q$-deformation stretches or compresses the tail of the distribution to absorb the finite-size skewness, preserving the nonnegativity of the probability measure across the entire macroscopic state space.

\begin{example}[Global Nonnegativity and Quantitative Evaluation in Continuous Asymmetric Fluctuations]
To computationally demonstrate this theoretical property and to provide a quantitative evaluation, we examine a continuous asymmetric system: the sum of $n$ independent exponentially distributed variables.
This follows a Gamma distribution, providing a non-Gaussian baseline with variance $V=1$ and skewness $T=2$.

As established in Theorem \ref{thm:absorption}, the dynamic scaling law dictates $q_n = 1 - \frac{T}{3V^2 n} = 1 - \frac{2}{3n}$.
The resulting $q$-Gaussian probability density function $f_{q_n}(z)$ for the standardized macroscopic fluctuation $z$ is explicitly given by:
\begin{equation}
    f_{q_n}(z) = \frac{1}{Z_{q_n}} \left[ 1 - \frac{1-q_n}{3-q_n} z^2 \right]_+^{\frac{1}{1-q_n}},
    \label{eq:explicit_q_gaussian}
\end{equation}
where $[x]_+ := \max\{x, 0\}$ ensures the non-negativity of the probability measure.
The normalization constant for the compact support regime ($q_n < 1$) is analytically determined using the Beta function $B(u,v)$ as:
\begin{equation}
    Z_{q_n} = \left( \frac{3-q_n}{1-q_n} \right)^{\frac{1}{2}} B\left( \frac{2-q_n}{1-q_n}, \frac{1}{2} \right).
\end{equation}

Figure \ref{fig:continuous_gamma_panels} illustrates this exact analytical distribution alongside the standard Edgeworth expansion across systematically increasing sample sizes ($n=12, 36, 288$).
While all approximations roughly coincide in the bulk, the inset highlights the deep-tail regime.
The standard Edgeworth expansion crosses the horizontal axis, assigning strictly negative probability densities to these macroscopic states.
In contrast, the $q$-deformation geometrically restricts its support to a physically valid compact domain, avoiding negative probabilities.

\begin{figure}[htbp]
    \centering
    \includegraphics[width=\textwidth, height=0.85\textheight, keepaspectratio]{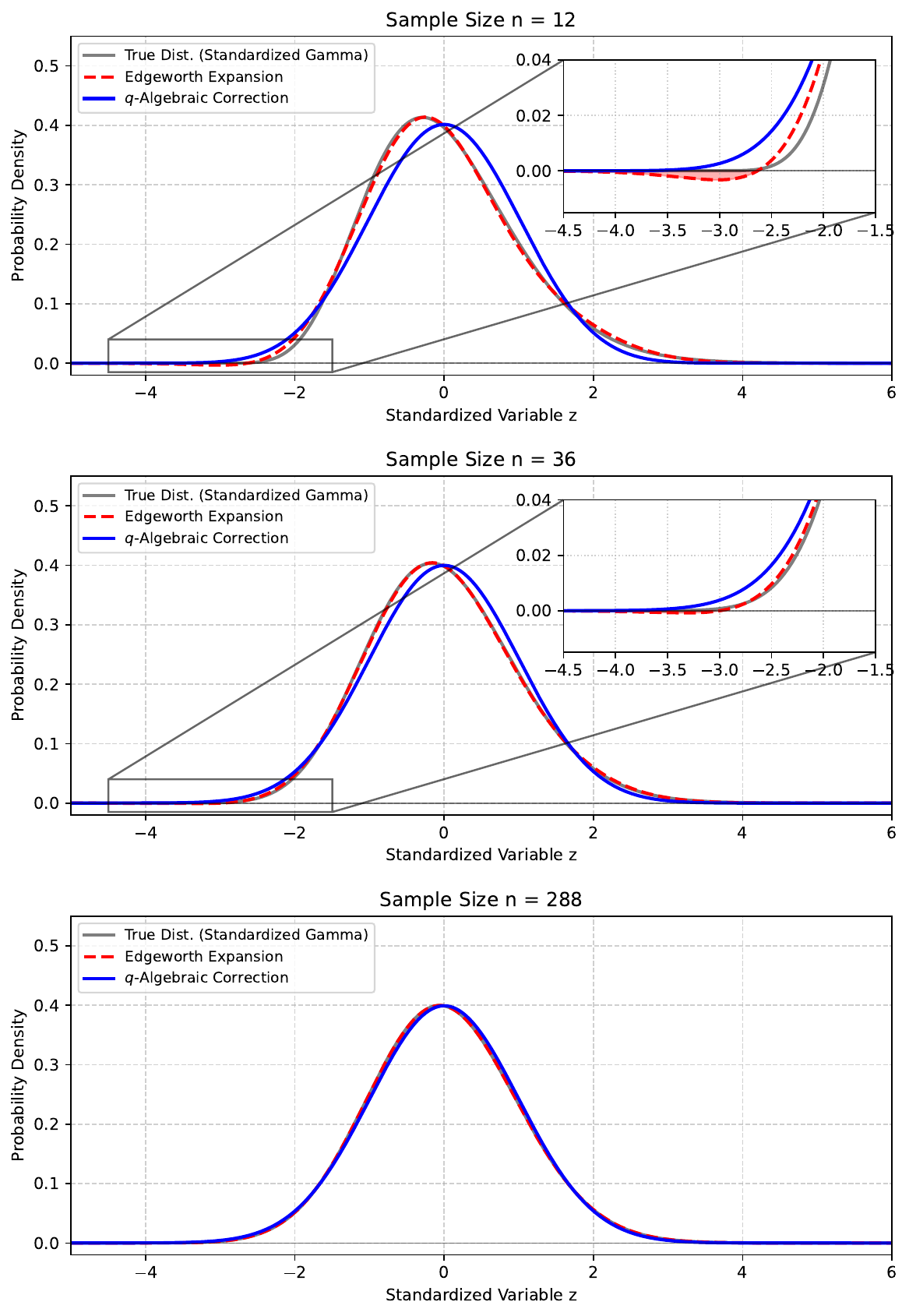}
    \caption{Probability density functions of the standardized macroscopic fluctuation for the sum of exponential variables across systematically increasing sample sizes ($n=12, 36, 288$).
The insets magnify the deep-tail regime, demonstrating that while the standard Edgeworth expansion (dashed red line) yields unphysical negative probabilities, the proposed $q$-algebraic correction (solid blue line) guarantees global nonnegativity through its compact support. For $n=288$, the approximations closely converge to the true distribution and the structural deviations in the deep tail become visually indistinguishable, thus the inset is omitted.}
    \label{fig:continuous_gamma_panels}
\end{figure}

To quantitatively evaluate the accuracy of the proposed $q$-algebraic correction against the standard Edgeworth expansion, we employ the Kolmogorov-Smirnov (KS) distance.
While the Kullback-Leibler (KL) divergence is a commonly used metric, it is mathematically undefined (diverges) in the present context because the Edgeworth expansion yields negative probabilities.
The KS distance avoids this singularity and provides a fair comparison across the entire domain.

Table \ref{tab:ks_distance} summarizes the KS distances for this standard Gamma distribution.
The Edgeworth expansion exhibits a slightly smaller KS distance, primarily because its asymmetric cubic term is specifically designed to tightly fit the bulk of the true right-skewed distribution.
In contrast, the $q$-Gaussian approximation, being a mathematically symmetric function, shows a slightly larger geometric deviation in the central region.

However, this minor deviation in the bulk is a necessary mathematical trade-off to preserve macroscopic physical consistency.
By algebraically absorbing the finite-size skewness into the nonextensivity parameter $q_n$, the exact algebraic absorption eliminates the critical breakdown---the emergence of negative probabilities.
Thus, while the Edgeworth expansion prioritizes local curve-fitting at the cost of global physical validity, the $q$-algebraic framework guarantees the strict non-negativity of probabilities and the structural stability of the macroscopic fluctuation scaling.

\begin{table}[htbp]
    \centering
    \caption{Quantitative comparison of the Kolmogorov-Smirnov (KS) distance between the true distribution (standardized Gamma) and the approximated distributions at various sample sizes ($n$).}
    \label{tab:ks_distance}
    \begin{tabular}{ccc}
        \hline \hline
        Sample Size ($n$) & Edgeworth Expansion & $q$-Algebraic Correction \\
        \hline
        $12$  & $0.006500$ & $0.038532$ \\
        $36$  & $0.001975$ & $0.022192$ \\
        $288$ & $0.000229$ & $0.007837$ \\
        \hline \hline
    \end{tabular}
\end{table}
\end{example}

\subsection{Global Stability and Thermodynamic Limits}
The necessity of the $q$-parameter in i.i.d. systems therefore lies in its global stability.
While the Edgeworth expansion is an asymptotic mathematical approximation, the $q_n$-scaled framework behaves as a structurally self-consistent physical model.
By replacing perturbative polynomials with generalized algebra, we extend the finite-size fluctuation bounds from localized approximations to globally valid thermodynamic limits.
This global stability demonstrates that the nonextensive algebraic framework provides a mathematical foundation for describing the complete distribution of fluctuations in finite physical systems.


\section{Higher-Order Asymptotics and Structural Equivalence}
\label{sec:higher_order}

In the previous sections, we established that the dynamic scaling $1-q_n = \alpha n^{-1}$ coupled with the specific coefficient $\alpha$ exactly absorbs the third-order macroscopic skewness.
The theoretical implications of this dynamic scaling extend beyond a localized third-order correction.
In this section, we demonstrate that this simple scaling law acts as an algebraic generator that matches the entire asymptotic hierarchy of arbitrary higher-order macroscopic fluctuations.

\subsection{The Generalized $O(n^{1-k/2})$ Correspondence}

\begin{theorem}[Algebraic Correspondence of Asymptotic Orders]
\label{thm:resonance}
Let the centralized macroscopic fluctuation be $W_n = O(n^{1/2})$ and the generalized parameter scale as $1-q_n = O(n^{-1})$.
The $k$-th degree term of the $q$-logarithmic algebraic expansion matches the asymptotic order of the $(k+1)$-th moment correction in the classical perturbative expansion (Edgeworth expansion), yielding the universal order $O(n^{1-k/2})$.
\end{theorem}

\begin{proof}
As established in Eq.~\eqref{eq:q_fluctuation_def}, the $q$-deformation of the centralized fluctuation $W_n$ expands generally via the $q$-deformed algebraic structure as:
\begin{equation}
    s_{q_n}(X^n) = n S_{\text{BG}} + W_n + \sum_{k=2}^{\infty} \frac{(1-q_n)^{k-1}}{k!} (W_n^k - \mathbb{E}[W_n^k]).
\label{eq:general_higher_order_expansion}
\end{equation}
By substituting the standard growth rate of the normal fluctuation $W_n = O(n^{1/2})$ and the proposed dynamic scaling limit $1-q_n = O(n^{-1})$ into the $k$-th degree term, the asymptotic order evaluates to:
\begin{equation}
    O\left( (n^{-1})^{k-1} \right) \times O\left( (n^{1/2})^k \right) = O\left( n^{-k+1} \right) \times O\left( n^{k/2} \right) = O\left( n^{1-k/2} \right).
\label{eq:general_order_resonance}
\end{equation}

Evaluating \eqref{eq:general_order_resonance} for successive values of $k$ reveals a structural correspondence with the classical finite-size limits:
\begin{itemize}
    \item $k=1 \implies O(n^{1/2})$ : Recovers the standard Central Limit Theorem scale (Gaussian variance).
    \item $k=2 \implies O(1)$ : Generates the third-order skewness correction, which we have proven to be fully absorbed via $\alpha = \gamma_1 / (3\sigma)$.
    \item $k=3 \implies O(n^{-1/2})$ : Matches the asymptotic order of the fourth-order kurtosis penalty.
    \item $k=m \implies O(n^{1-m/2})$ : Aligns with the corresponding $(m+1)$-th moment correction order for any integer $m \ge 1$.
\end{itemize}
\end{proof}

\subsection{Beyond the Third Order: Towards Exact Resummation}
The generalized order relation $O(n^{1-k/2})$ provides analytical evidence that the dynamic parameter $q_n$ is not a fitting parameter, but an algebraic generator for finite-size structural corrections.
The single geometric scaling rule $1-q_n = O(n^{-1})$ establishes an isomorphic relationship that offsets the macroscopic behavior of all higher-order moments without relying on the combinatorial complexity of accumulating orthogonal Hermite polynomials.

While this paper has mathematically proven the analytical absorption of the macroscopic fluctuations up to the third order (skewness) by matching the specific algebraic coefficients, the automatic emergence of the $O(n^{-1/2})$ scale for $k=3$ (and beyond) suggests a more universal paradigm.
It implies that the $q$-deformed framework structurally absorbs the kurtosis and infinite higher-order deviations.

Elucidating a structural mapping between the classical orthogonal polynomials and the $q$-logarithmic Taylor coefficients for arbitrary higher orders ($k \ge 3$) remains an open mathematical challenge.
Resolving this problem extends beyond a theoretical coefficient matching; it has the potential to establish the $q$-deformed algebraic structure as an exact, globally stable resummation technique for divergent asymptotic expansions in statistical mechanics.


\section{Conclusion}
\label{sec:conclusion}

In this paper, we proposed a structural resolution to the physical breakdown of perturbative finite-size approximations by introducing a dynamically scaled $q$-deformed framework.
Traditional methods in statistical mechanics, such as the Edgeworth expansion, attempt to capture non-Gaussian macroscopic fluctuations by appending additive polynomial corrections to a Gaussian baseline.
However, as demonstrated, this leads to unphysical negative probabilities in the deep-tail (large deviation) regime.
By treating finite-size skewness not as an external error penalty but as an algebraic deformation of the probability measure, we established a globally stable framework that preserves physical nonnegativity.

Our analytical results demonstrate that applying the dynamic scaling law $1-q_n = O(n^{-1})$ to the macroscopic fluctuation absorbs the third-order finite-size effects.
Furthermore, we proved that this single geometric scaling acts as a universal algebraic generator: the $k$-th degree term of the $q$-logarithmic expansion matches the asymptotic order $O(n^{1-k/2})$ of the classical $(k+1)$-th moment correction.
This structural matching implies that the nonextensive algebraic structure resolves the combinatorial explosion of orthogonal polynomials, functioning as a structurally stable resummation of divergent asymptotic expansions.

Finally, this work focused on independent and identically distributed (i.i.d.) sequences where the finite-size deformation parameter asymptotically vanishes ($q_n \to 1$ as $n \to \infty$).
Within this regime, the established algebraic framework offers a mathematical foundation for the asymptotic analysis of large deviations, demonstrating that nonextensive structures intrinsically govern finite-size fluctuations even in standard macroscopic limits.



\section*{Acknowledgements}
This work was supported by JSPS KAKENHI Grant Number 26K14703. The author thanks the anonymous reviewers for their constructive comments, which contributed to the structural clarification of this manuscript.





\end{document}